\documentclass[a4paper,11pt]{article}
\usepackage{amsmath,amssymb,amsthm,graphicx}
\usepackage[top=80pt, bottom=90pt, left=70pt, right=70pt]{geometry}
\usepackage{comment}
\usepackage{color}
\usepackage{latexsym}
\def\qed{\hfill $\Box$}

\newcommand{\argmin}{\operatornamewithlimits{argmin}}

\newtheorem{thm}{\bfseries Theorem}
\newtheorem{lem}[thm]{\bfseries Lemma} 
\newtheorem{remark}[thm]{\bfseries Remark} 

\usepackage{bm}

\begin{document}
\title{Discrete Convexity in Joint Winner Property}
\author{Yuni Iwamasa\footnotemark[1] \and Kazuo Murota\footnotemark[2] \and Stanislav \v{Z}ivn\'y\footnotemark[3]}
\date{\today}
\footnotetext[1]{Department of Mathematical Informatics,
Graduate School of Information Science and Technology,
University of Tokyo, Tokyo, 113-8656, Japan.\\
Email: \texttt{yuni\_iwamasa@mist.i.u-tokyo.ac.jp}}
\footnotetext[2]{Department of Business Administration,
	Tokyo Metropolitan University, Tokyo, 192-0397, Japan.\\
	Email: \texttt{murota@tmu.ac.jp}}
\footnotetext[3]{Department of Computer Science,
	University of Oxford, Oxford, OX1 3QD, United Kingdom.\\
	Email: \texttt{standa.zivny@cs.ox.ac.uk}}
\maketitle

\begin{abstract}
	In this paper, we reveal a relation between joint winner property (JWP) in the field of valued constraint satisfaction problems (VCSPs) and M${}^\natural$-convexity in the field of discrete convex analysis (DCA).
	We introduce the M${}^\natural$-convex completion problem,
	and show that a function $f$ satisfying the JWP is Z-free if and only if
	a certain function $\overline{f}$ associated with $f$ is M${}^\natural$-convex completable.
	This means that if a function is Z-free,
	then the function can be minimized in polynomial time via M${}^\natural$-convex intersection algorithms.
	Furthermore we propose a new algorithm for Z-free function minimization,
	which is faster than previous algorithms for some parameter values.
\end{abstract}
\begin{quote}
\end{quote}

\section{Introduction}
A {\it valued constraint satisfaction problem (VCSP)} is a general framework for discrete optimization (see \cite{book/Zivny12} for details).
Informally, the VCSP framework deals with the minimization problem of a function represented as the sum of ``small'' arity functions.
It is known that various kinds of combinatorial optimization problems can be formulated in the VCSP framework.
In general, the VCSP is NP-hard.
An important line of research is to investigate which classes of instances are solvable in polynomial time,
and why these classes ensure polynomial time solvability.
Cooper--\v{Z}ivn\'{y}~\cite{AI/CZ11} showed that if a function represented as the sum of unary or binary functions satisfies the {\it joint winner property (JWP)},
then the function can be minimized in polynomial time.
This gives an example of a class of instances that are solvable in polynomial time.

In this paper,
we present the reason why JWP ensures polynomial time solvability via {\it discrete convex analysis (DCA)}~\cite{book/Murota03}, particularly, {\it M${}^\natural$-convexity}~\cite{MOR/MS99}.
DCA is a theory of convex functions on discrete structures,
and M${}^\natural$-convexity is one of the important convexity concepts in DCA.
M${}^\natural$-convexity appears in many areas such as operations research, economics, and game theory (see e.g.,~\cite{book/Murota03, incollection/M09,  JMID/M16}).

The results of this paper are summarized as follows:
\begin{itemize}
	\item We reveal a relation between JWP and M${}^\natural$-convexity.
	That is, we give a DCA interpretation of polynomial-time solvability of JWP.
	\item To describe the connection of JWP and M${}^\natural$-convexity,
	we introduce the {\it M${}^\natural$-convex completion problem},
	and give a characterization of M${}^\natural$-convex completability.
	\item By utilizing a DCA interpretation of JWP,
	we propose a new algorithm for Z-free function minimization,
	which is faster than previous algorithms for some parameter values.
\end{itemize}
This study will hopefully be the first step towards fruitful interactions between VCSPs and DCA.

\paragraph{Notations.}
Let $\mathbf{R}$ and $\mathbf{R}_+$ denote the sets of reals and nonnegative reals, respectively.
In this paper, functions can take the infinite value $+ \infty$, where $a < + \infty$, $a + \infty = + \infty$ for $a \in \mathbf{R}$, and $0 \cdot (+\infty) = 0$.
Let $\overline{\mathbf{R}} := \mathbf{R} \cup \{+ \infty\}$ and $\overline{\mathbf{R}}_+ := \mathbf{R}_+ \cup \{+ \infty\}$.
For a function $f : \{0,1\}^n \rightarrow \overline{\mathbf{R}}$, the effective domain is denoted as ${\rm dom}\ f := \{ x \in \{0,1\}^n \mid f(x) < + \infty \}$.
For a positive integer $k$, we define $[k] := \{1,2,\dots,k\}$.
For $x = (x_1,x_2,\dots,x_n) \in \mathbf{R}^n$, we define $\textrm{supp}^+(x) := \{ i \in [n] \mid x_i > 0 \}$.

\section{Preliminaries}
\paragraph{Joint Winner Property.}
Let $d_i \geq 2$ be a positive integer and $D_i := [d_i]$ for $i \in [r]$.
We consider a function $f : D_1 \times D_2 \times \cdots \times D_r \rightarrow \overline{\mathbf{R}}_+$ represented as the sum of unary or binary functions as
\begin{align}\label{eq:binary VCSP}
f(x_1,x_2,\dots,x_r) = \sum_{i \in [r]} c_i(x_i) + \sum_{1\leq i < j \leq r}c_{ij}(x_i, x_j),
\end{align}
where $c_i : D_i \rightarrow \mathbf{R}_+$ is a unary function for $i \in [r]$
and $c_{ij} : D_i \times D_j \rightarrow \overline{\mathbf{R}}_+$ is a binary function for $1 \leq i < j \leq r$.
Furthermore we assume $c_{ij} = c_{ji}$ for distinct $i,j \in [r]$.
A function $f$ of the form (\ref{eq:binary VCSP}) is said~\cite{AI/CZ11} to satisfy the {\it joint winner property (JWP)} if it holds that
\begin{align}\label{eq:JWP}
c_{ij}(a,b) \geq \min\{c_{jk}(b,c), c_{ik}(a, c)\}
\end{align}
for all distinct $i,j,k \in [r]$ and all $a \in D_i, b \in D_j, c \in D_k$.
A function $f$ of the form (\ref{eq:binary VCSP}) satisfying the JWP is said to be {\it Z-free}
if it satisfies that
\begin{align}\label{eq:Z-free}
	|\argmin \{ c_{ij}(a,c), c_{ij}(a,d),c_{ij}(b,c), c_{ij}(b,d) \}| \geq 2
\end{align}
for any $i,j \in [r]$ ($i \neq j$), $\{a,b\} \subseteq D_i$ ($a \neq b$),
and $\{c,d\} \subseteq D_j$ ($c \neq d$).

Cooper--\v{Z}ivn\'{y}~\cite{AI/CZ11} showed that if $f$ of the form (\ref{eq:binary VCSP}) satisfies the JWP, then $f$ can be minimized in polynomial time.
In fact,
they showed that if $f$ satisfies the JWP,
then $f$ can be transformed into a certain Z-free function $f'$ in polynomial time such that a minimizer of $f'$ is also a minimizer of $f$.
Moreover they showed that a Z-free function can be minimized in polynomial time.

JWP appears in many contexts.
For example, JWP identifies a tractable class of the MAX-2SAT problem,
which is a well-known NP-hard problem~\cite{TCS/GJS76}.
Indeed, for a 2-CNF formula $\psi$,
we can represent the MAX-2SAT problem for $\psi$ as a Boolean binary $\{0,1\}$-valued VCSP instance.
In this binary VCSP instance,
JWP is equivalent to the following condition on $\psi$:
if clauses $(x_1 \vee x_2)$ and $(x_1 \vee x_3)$ are contained in $\psi$,
then so is $(x_2 \vee x_3 )$.
In addition, the {\sc AllDifferent} constraint~\cite{AAAI/R94} and {\sc SoftAllDiff} constraint~\cite{CP/PRB01}
can be regarded as special cases of the JWP,
and certain scheduling problems introduced in~\cite{CACM/BCS74, OR/H73} satisfy the JWP.
See also Examples~5--8 in \cite{AI/CZ11} for details.

\paragraph{M${}^\natural$-Convexity.}
A function $f : \{0,1\}^n \rightarrow \overline{\mathbf{R}}$ is said~\cite{book/Murota03,incollection/M09} to be {\it M${}^\natural$-convex} if for all $x,y \in \{0,1\}^n$ and all $i \in \textrm{supp}^+(x-y)$ there exists $j \in {\rm supp}^+(y-x) \cup \{0\}$ such that
\begin{align}\label{eq:M natural def}
f(x) + f(y) \geq f(x - \chi_i + \chi_j) + f(y + \chi_i - \chi_j),
\end{align}
where $\chi_i$ is the $i$th unit vector and $\chi_0$ is the zero vector.
A function $f : \{0,1\}^n \rightarrow \overline{\mathbf{R}}$ is said~\cite{book/Murota03} to be {\it M${}^\natural_2$-convex} if $f$ can be represented as the sum of two M${}^\natural$-convex functions.
It is well known that M${}^\natural$-convex functions can be minimized in polynomial time.
Furthermore if we are given two M${}^\natural$-convex functions $g$ and $h$,
we can minimize an M${}^\natural_2$-convex function $f = g+h$ in polynomial time by solving the so-called ``M${}^\natural$-convex intersection problem.''

\begin{thm}[{\cite[Theorem 10]{ET/RGP02}\cite[Theorem 6.5]{JORSJ/ST15}; see also~\cite[Theorem~3.3]{JMID/M16}}]\label{thm:{0,1}M}
	A function $f : \{0,1\}^n \rightarrow \overline{\mathbf{R}}$ with the zero vector in ${\rm dom}\ f$ is M${}^{\natural}$-convex if and only if $f$ satisfies the following two conditions:
	\begin{description}
		\item[Condition 1:] For all distinct $i,j,k \in [n]$ and all $ z \in \{0, 1\}^n$ with ${\rm supp}^+(z) \subseteq [n] \setminus \{i,j,k\}$, it holds that
		\begin{align}
		f(z + \chi_i + \chi_j) + f(z + \chi_k) \geq \min\{f(z + \chi_j + \chi_k) + f(z + \chi_i), f(z + \chi_i + \chi_k) + f(z + \chi_j)\}.
		\end{align}
		\item[Condition 2:] For all distinct $i,j \in [n]$ and all $z \in \{0, 1\}^n$ with ${\rm supp}^+(z) \subseteq [n] \setminus \{i,j\}$, it holds that
		\begin{align*}
		f(z + \chi_i + \chi_j) + f(z) \geq f(z + \chi_i) + f(z + \chi_j).
		\end{align*}
	\end{description}
\end{thm}

We pay special attention to quadratic M${}^\natural$-convex functions.
Using Theorem~\ref{thm:{0,1}M}, we provide a necessary and sufficient condition for the M${}^\natural$-convexity of a function $f : \{0,1\}^n \rightarrow \overline{\mathbf{R}}$ of the form
\begin{align}\label{eq:f}
f(x_1,x_2, \dots, x_n) := \sum_{i \in [n]} h_i x_i + \sum_{1 \leq i < j \leq n} h_{ij} x_i x_j \qquad \left((x_1,x_2,\dots,x_n) \in \{0,1\}^n \right),
\end{align}
where we assume $h_{ij} = h_{ji}$ and $h_i < + \infty$ for $i,j \in [n]$.
\begin{lem}\label{lem:Mnatural {0,1}}
	A function $f$ of the form {\rm (\ref{eq:f})} is M${}^{\natural}$-convex if and only if it satisfies the following:
	\begin{itemize}
		\item $h_{ij} \geq \min \{h_{ik}, h_{jk}\} \qquad (i,j,k \text{ : distinct})$.
		\item $h_{ij} \geq 0 \qquad (i,j \text{ : distinct})$.
	\end{itemize}
\end{lem}
In Lemma~\ref{lem:Mnatural {0,1}}, $h_{ij}$ can take the infinite value $+\infty$,
whereas all $h_{ij}$'s are assumed to be finite in the characterization in~\cite{JJIAM/HM04} and~\cite{incollection/M09}.
In particular, we refer to the first condition $h_{ij} \geq \min \{h_{ik}, h_{jk}\}$ ($i,j,k$ : distinct) as the {\it anti-ultrametric property}.
Note that no conditions are imposed on $h_i$.
The proof of Lemma~\ref{lem:Mnatural {0,1}} is in Section~\ref{sec:proofs}

By Lemma~\ref{lem:Mnatural {0,1}},
we know that M${}^\natural$-convexity of a function of the form (\ref{eq:f}) depends only on quadratic coefficients $(h_{ij})_{i,j \in [n]}$.
We say that a function $f$ of the form (\ref{eq:f}) is defined by $(h_{ij})_{i,j \in [n]}$
if the quadratic coefficients of $f$ is equal to $(h_{ij})_{i,j \in [n]}$.

\section{M${}^\natural$-Convexity in Joint Winner Property}
\paragraph{M${}^\natural$-Convex Completion Problem.}
We introduce the {\it M${}^\natural$-convex completion problem},
and give a characterization of an M${}^\natural$-convex completable function on $\{0,1\}^n$ defined by $(h_{ij})_{i,j \in [n]}$.
The M${}^\natural$-convex completion problem is the following:
\begin{description}
	\setlength{\itemsep}{-2pt}
	\item[Given:] $(h_{ij})_{i,j \in [n]}$ such that $h_{ij} \in \overline{\mathbf{R}}$ or $h_{ij}$ is undefined for every distinct $i,j \in [n]$.
	\item[Question:] By assigning appropriate values in $\overline{\mathbf{R}}$ to ``undefined" elements of $(h_{ij})_{i,j \in [n]}$,
	can we construct an M${}^\natural$-convex function $f : \{0,1\}^n \rightarrow \overline{\mathbf{R}}$ of the form (\ref{eq:f})?
\end{description}
It should be clear that a defined element can be equal to $+\infty$
and the infinite value ($+\infty$) may be assigned to undefined elements.
If there is an appropriate assignment of $(h_{ij})_{i,j \in [n]}$,
then $(h_{ij})_{i,j \in [n]}$ is said to be {\it M${}^\natural$-convex completable}.
If $h_{ij} < 0$ or $h_{ij} < \min\{h_{jk}, h_{ik} \}$ holds for some defined elements $h_{ij}, h_{jk}, h_{ik}$,
then we obviously know that $(h_{ij})_{i,j \in [n]}$ is not M${}^\natural$-convex completable.
Hence in considering the M${}^\natural$-convex completion problem,
we assume that
\begin{align}
&h_{ij} \geq 0,\label{eq:assumption h_ij}\\
&h_{ij} \geq \min\{h_{jk}, h_{ik} \}\label{eq:assumption h_ij h_jk h_ik}
\end{align}
for all defined elements $h_{ij}, h_{jk}, h_{ik}$.

For quadratic coefficients $H := (h_{ij})_{i,j \in [n]}$ containing undefined elements,
we define the {\it assignment graph} of $H$ as a graph $G_H = ([n], E_H; w)$,
where $E_H := \{ \{i,j\} \mid \text{$i \neq j$ and $h_{ij}$ is defined} \}$
and $w : E_H \rightarrow \overline{\mathbf{R}}_+$ is defined by $w(\{i,j\}) := h_{ij}$ for $\{i,j\} \in E_H$.
Then the following theorem holds.
\begin{thm}\label{thm:Mnatural conv completable}
	$H := (h_{ij})_{i,j \in [n]}$ is M${}^\natural$-convex completable if and only if
	$|\argmin_{e \in C} w(e) | \geq 2$ holds for every chordless cycle $C$ of $G_H$.
\end{thm}
The proof of Theorem~\ref{thm:Mnatural conv completable} is in Section~\ref{sec:proofs}.

\begin{remark}\label{rem:completion}
	\upshape
	Farach--Kannan--Warnow~\cite{Algo/FKW95} introduced the {\it matrix sandwich problem for ultrametric property},
	which contains the M${}^\natural$-convex completion problem as a special case.
	They also constructed an $O(m + n\log n)$-time algorithm for the matrix sandwich problem for ultrametric property,
	where $m$ is the number of defined elements.
	In our setting, $m$ = $O(n^2)$.
	Hence, by using this algorithm,
	we can obtain an appropriate M${}^\natural$-convex completion in $O(n^2)$ time if one exists.
	An $O(n^2)$-time algorithm based on Farach--Kannan--Warnow's algorithm is the following:
	Suppose that all $h_{ij}$ are finite
	(if there exists $h_{ij}$ with $h_{ij} = +\infty$,
	then we can redefine the value of $h_{ij}$ as a sufficiently large finite value $M$).
	Take any maximum forest $F$ of $G_H$.
	Let $\alpha_1 > \alpha_2 > \cdots > \alpha_p$ be the distinct values of defined elements of $(h_{ij})_{\{i,j\} \in F}$.
	For $k = 1, \dots, p-1$,
	let $F^{\alpha_k}$ be the subgraph of $F$ induced by the edges with weight at least $\alpha_k$,
	i.e., $F^{\alpha_k} := \{ \{i,j\} \in F \mid h_{ij} \geq \alpha_k \}$.
	Then, for each $\{i,j\} \not\in E_H$ with $i,j$ connected in $G_H$, set $h_{ij}$ to $\alpha_k$,
	where $k$ is the minimum number such that $i,j$ is connected in $F^{\alpha_k}$.
	For each $\{ i,j \} \not\in E_H$ with $i, j$ disconnected in $G_H$,
	set $h_{ij}$ to $\alpha_{p}$.
\end{remark}

In this paper,
we present a graphic characterization of M${}^\natural$-convex completability.
With this characterization,
we provide a DCA interpretation of polynomial-time solvability of JWP.

\paragraph{Transformation into a Function over $\{0,1\}$.}
To connect JWP and M${}^\natural$-convexity,
we introduce a transformation of a function $f : D_1 \times D_2 \times \cdots \times D_r \rightarrow \overline{\mathbf{R}}$ into a function $\hat{f} : \{0, 1\}^U \rightarrow \overline{\mathbf{R}}$,
where $U$ is the set of all assignments to variables, that is,
\begin{align*}
U := \{(1,1),(1,2), \dots, (1,d_1),(2,1),(2,2), \dots, (2,d_2),\dots,(r,1),(r,2), \dots, (r, d_r)\}.
\end{align*}
We consider the following correspondence between $x = (x_1, x_2, \dots, x_r) \in D_1 \times D_2 \times \dots \times D_r$ and $\hat{x} = (\hat{x}_{(1,1)}, \dots, \hat{x}_{(1,d_1)}, \hat{x}_{(2,1)}, \dots, \hat{x}_{(2,d_2)}, \dots, \hat{x}_{(r,1)}, \dots, \hat{x}_{(r,d_r)}) \in \{0,1\}^U$:
\begin{align}\label{eq:D to {0,1}}
(x_1, x_2, \dots, x_r) \mapsto (\underbrace{0, \dots, 0,\overset{(1,x_1)}{\check{1}},0, \dots, 0}_{d_1}, \underbrace{0, \dots,0, \overset{(2,x_2)}{\check{1}},0, \dots, 0}_{d_2}, \dots, \underbrace{0, \dots, 0,\overset{(r,x_r)}{\check{1}},0, \dots, 0}_{d_r}). 
\end{align}
That is, $\hat{x}_{(i,a)} = 1$ means that we assign $a$ to $x_i$,
and $\hat{x}_{(i,a)} = 0$ means that we do not.
In view of~(\ref{eq:D to {0,1}}), define a function $\hat{f}$ by
\begin{align*}
\hat{f}(\hat{x}) := \begin{cases} f(x) & \text{if there exists $x$ satisfying~(\ref{eq:D to {0,1}})}, \\ +\infty & \text{otherwise} \end{cases}
 \qquad (\hat{x} \in \{0, 1\}^U).
\end{align*}
Note that minimizing $f$ is equivalent to minimizing $\hat{f}$.

Now we consider the transformation of $f$ of the form (\ref{eq:binary VCSP}) into $\hat{f}$,
where $f$ is given in terms of $c_i$ for $i \in [r]$ and $c_{ij}$ for $i,j \in [r]$.
We define $\overline{f} : \{0,1\}^U \rightarrow \overline{\mathbf{R}}_+$ by
\begin{align}\label{eq:overline f}
\overline{f}(\hat{x}) := \sum_{(i,a) \in U} c_i(a)\hat{x}_{(i,a)} + \sum_{(i,a), (j,b) \in U,\ (i,a) \neq (j,b)} h_{(i,a),(j,b)}\hat{x}_{(i,a)} \hat{x}_{(j,b)} \qquad (\hat{x} \in \{0, 1\}^U),
\end{align}
where
\begin{align}\label{eq:h_(i,a)(j,b)}
h_{(i,a), (j,b)} := \begin{cases} c_{ij}(a, b) & \text{if $i \neq j$},\\
\text{undefined} & \text{if $i = j$}. \end{cases}
\end{align}
We also define $\delta_U : \{0,1\}^U \rightarrow \overline{\mathbf{R}}$ by
\begin{align*}
\delta_U(\hat{x}) :=
\begin{cases}
0 & \text{if there exists $x$ satisfying~(\ref{eq:D to {0,1}})},
\\ +\infty & \text{otherwise}
\end{cases} \qquad (\hat{x} \in \{0, 1\}^U),
\end{align*}
which is the indicator function for the feasible assignments.
Then we have
\begin{align*}
\hat{f}(\hat{x}) = \overline{f}(\hat{x}) + \delta_U(\hat{x}) \qquad (\hat{x} \in \{0,1\}^U),
\end{align*}
where arbitrary values in $\overline{\mathbf{R}}$ may be assigned to the undefined elements $h_{(i,a), (i,b)}$ in $\overline{f}$
without affecting the value of $\hat{f}$.
Indeed, if $\hat{x} \in \textrm{dom }\delta_U$,
then $\hat{x}_{(i,a)} \hat{x}_{(i, b)} = 0$ for all $i \in [r]$ and all distinct $a, b \in D_i$.
Hence $h_{(i,a), (i,b)} \hat{x}_{(i,a)} \hat{x}_{(i,b)} = 0$ holds for each undefined element $h_{(i,a), (i,b)}$ by the definition~(\ref{eq:h_(i,a)(j,b)}).
In particular, the set of minimizers of both $\hat{f}(\hat{x})$ and $\overline{f}(\hat{x}) + \delta_U(\hat{x})$ are the same.

It is clear that $\delta_U$ is M${}^\natural$-convex
(${\rm dom}\ \delta_U$ is the base family of a partition matroid, which is a direct sum of matroids of rank 1).
Hence if $(h_{(i,a), (j,b)})_{(i,a), (j,b) \in U}$ has an M${}^\natural$-convex completion $(\tilde{h}_{(i,a), (j,b)})_{(i,a), (j,b) \in U}$,
then $\overline{f}$ defined by $(\tilde{h}_{(i,a), (j,b)})_{(i,a), (j,b) \in U}$ is M${}^\natural$-convex and $\hat{f} = \overline{f} + \delta_U$ is M${}^\natural_2$-convex.
This means that $\hat{f}$ can be minimized in polynomial time.
We need the values of $(\tilde{h}_{(i,a), (j,b)})_{(i,a), (j,b) \in U}$
in a minimization algorithm of M${}^\natural_2$-convex functions.

A function of the form (\ref{eq:binary VCSP}) satisfies the JWP if and only if
$h_{(i,a),(j,b)} \geq \min\{ h_{(j,b), (k,c)}, h_{(i,a),(k,c)} \}$ holds for defined elements $h_{(i,a),(j,b)}, h_{(j,b), (k,c)}, h_{(i,a),(k,c)}$ given in~(\ref{eq:overline f}).
Hence $(h_{(i,a), (j,b)})_{(i,a), (j,b) \in U}$ satisfies the assumptions (\ref{eq:assumption h_ij}) and (\ref{eq:assumption h_ij h_jk h_ik}) for the M${}^\natural$-convex completion problem.
Theorem~\ref{thm:Mnatural conv completable} implies
the following theorem (the proof is in Section~\ref{sec:proofs}).
\begin{thm}\label{thm:main}
	For a function $f$ of the form~{\rm (\ref{eq:binary VCSP})},
	let $(h_{(i,a), (j,b)})_{(i,a), (j,b) \in U}$ be defined by~{\rm (\ref{eq:h_(i,a)(j,b)})}.
	Then $(h_{(i,a), (j,b)})_{(i,a), (j,b) \in U}$ is M${}^\natural$-convex completable
	if and only if $f$ (has the JWP and) is Z-free.
\end{thm}

\section{Algorithm}
By using a general algorithm for the
M${}^\natural$-convex intersection (minimization of M${}^\natural_2$-convex functions),
we can minimize Z-free functions of the form~(\ref{eq:binary VCSP}) in polynomial time.
Suppose that we are given $c_i : D_i \rightarrow \mathbf{R}_+$ for $i \in [r]$, $c_{ij} : D_i \times D_j \rightarrow \overline{\mathbf{R}}_+$ for $1 \leq i < j \leq r$,
and a Z-free function $f$ defined as~(\ref{eq:binary VCSP}).
We can minimize $f$ by minimizing 
$\hat{f} = \overline{f} + \delta_U$
with an M${}^\natural$-convex intersection algorithm.

Here we take advantage of the fact that
all the vectors in $\textrm{dom }\delta_U$
have a constant component sum, i.e.,
$\sum_{(i,a) \in U} \hat{x}_{(i,a)} = r$
for all $\hat{x} \in \textrm{dom }\delta_U$.
This implies
that 
$\delta_U$ is an M-convex function~\cite{book/Murota03}
and we can use an {\it M-convex intersection} algorithm.
An M-convex intersection algorithm is easier to describe than an M${}^\natural$-convex intersection algorithm,
though the time complexity is the same.
Therefore we devise a minimization algorithm for Z-free functions via an M-convex intersection algorithm.
Since the functions are defined on $\{0,1\}^n$, the proposed algorithm is actually a variant of valuated matroid intersection algorithms~\cite{book/Murota00}.
Specifically, let $\overline{f}|_r$ denote the restriction of 
$\overline{f}$ to the hyperplane containing $\textrm{dom }\delta_U$, i.e.,
\begin{align*}
\overline{f}|_r(\hat{x}) :=
\begin{cases}
\displaystyle
\overline{f}(\hat{x}) 
&
\text{if $\displaystyle\sum_{(i,a) \in U} \hat{x}_{(i,a)} = r$},\\
+\infty & \text{otherwise}.
\end{cases}
\end{align*}
Then 
minimizing $\overline{f} + \delta_U$
is equivalent to minimizing 
$\overline{f}|_r + \delta_U$,
where $\overline{f}|_r$ and $\delta_U$ are M-convex functions.

The proposed algorithm consists of three steps.

\begin{description}
	\item[Step 1:] 
	On the basis of Theorem \ref{thm:main},
	we construct an M${}^\natural$-convex function  
	$\overline{f} : \{0,1\}^U \rightarrow \overline{\mathbf{R}}$ 
	in~(\ref{eq:overline f})
	through an M${}^\natural$-convex completion
	of $(h_{(i,a),(j,b)})_{(i,a),(j,b) \in U}$ in~(\ref{eq:h_(i,a)(j,b)}).
	
	\item[Step 2:] We find a minimizer of $\overline{f}|_r$,
	to be used as an initial solution in Step 3.

	\item[Step 3:] 
	We find a minimizer of 
	$\overline{f}|_r + \delta_U$
	by the successive shortest path algorithm with potentials for 
	the M-convex intersection~\cite{book/Murota03} 
	(see also~\cite[Section~5.2]{book/Murota00}).
\end{description}

In Step 3 of the algorithm,
we use the
{\it auxiliary graph} 
$G_{\hat{x}, \hat{y}} = (V, E_{\hat{x}, \hat{y}})$ 
defined 
for $\hat{x} \in \textrm{dom }\overline{f}|_r$ and $\hat{y} \in \textrm{dom }\delta_U$
by
\begin{align}
V &:= \{s,t\} \cup U,\label{eq:V}\\
E_{\hat{x}} &:= \{ ((i,a),(j,b)) \mid (i,a),(j,b) \in U,\ \hat{x} + \chi_{(j,b)} - \chi_{(i,a)} \in \textrm{dom }\overline{f}|_r \},\label{eq:E_x}\\
E_{\hat{y}} &:= \{ ((i,a),(j,b)) \mid (i,a),(j,b) \in U,\ \hat{y} + \chi_{(i,a)} - \chi_{(j,b)} \in \textrm{dom }\delta_U \},\label{eq:E_y}\\
E^+ &:= \{ (s, (i,a)) \mid (i,a) \in \textrm{supp}^+(\hat{x} - \hat{y}) \},\label{eq:E^+}\\
E^- &:= \{ ((j,b), t) \mid (j,b) \in \textrm{supp}^+(\hat{y} - \hat{x}) \},\label{eq:E^-}\\
E_{\hat{x}, \hat{y}} &:= E_{\hat{x}} \cup E_{\hat{y}} \cup E^+ \cup E^-
\end{align}
with the arc length function 
$\ell = \ell_{\hat{x}, \hat{y}} : E_{\hat{x}, \hat{y}} \rightarrow \mathbf{R}$ 
given by
\begin{align}
\ell (u,v)
&:=
\begin{cases}
\overline{f}|_r(\hat{x} + \chi_{v} - \chi_{u}) - \overline{f}|_r(\hat{x}) & \text{if $(u,v) \in E_{\hat{x}}$},\\
0 & \text{otherwise}.
\end{cases}\label{eq:w}
\end{align}
Note that, by the definition of $\delta_U$, we can also describe $E_{\hat{y}}$ as $E_{\hat{y}} = \{ ((i,a),(i,b)) \mid i \in [r],\ a, b \in D_i,\ (i,a) \not\in \textrm{supp}^+(\hat{y}),\ (i,b) \in \textrm{supp}^+(\hat{y}) \}$.

\begin{description}
	\item[Algorithm for Z-free function minimization:]
	\item[Step 1:] Find an M${}^\natural$-convex completion $(\tilde{h}_{(i,a), (j,b)})_{(i,a), (j,b) \in U}$ of $(h_{(i,a), (j,b)})_{(i,a), (j,b) \in U}$,
	and define 
	$\overline{f} : \{0,1\}^U \rightarrow \overline{\mathbf{R}}$ 
	by
	\begin{align*}
	\overline{f}(\hat{x}) =
	\sum_{(i,a) \in U} c_i(a)\hat{x}_{(i,a)} + \sum_{(i,a), (j,b) \in U,\ (i,a) \neq (j,b)} \tilde{h}_{(i,a),(j,b)}\hat{x}_{(i,a)} \hat{x}_{(j,b)} \qquad (\hat{x} \in \{0,1\}^U).
	\end{align*}
	
	\item[Step 2:] Let $\hat{x}^* \in \{0,1\}^U$ be the zero vector.
	While $\sum_{(i,a) \in U} \hat{x}^*_{(i,a)} < r$, do the following:
	
	\begin{description}
		\item[Step 2-1:] 
		Obtain $(i,a)^* \in \argmin \{ \overline{f}(\hat{x}^* + \chi_{(i,a)}) \mid (i,a) \in U \setminus \textrm{supp}^+(\hat{x}^*) \}$.
		\item[Step 2-2:] 
		$\hat{x}^* \leftarrow \hat{x}^* + \chi_{(i,a)^*}$.
	\end{description}
	\item[Step 3:] 
	Let $p : V \rightarrow \mathbf{R}$ be a potential defined by $p(v) := 0$ for $v \in \{s,t\} \cup U$.
	Take any $\hat{y}^* \in \textrm{dom }\delta_U$.
	While $\hat{x}^* \neq \hat{y}^*$, do the following:
	\begin{description}
		\item[Step 3-1:]
		Make the auxiliary graph $G_{\hat{x}^*, \hat{y}^*}$.
		Define the modified arc length 
		$\ell_{p} : E_{\hat{x}^*, \hat{y}^*} \rightarrow \mathbf{R}$ 
		by 
		$\ell_{p}(u,v) := \ell(u,v) + p(u) - p(v)$ 
		for $(u,v) \in E_{\hat{x}^*, \hat{y}^*}$.
		
		\item[Step 3-2:] 
		For each $v \in V$, compute the length $\Delta p(v)$ of 
		an
		$s$-$v$ shortest path in $G_{\hat{x}^*, \hat{y}^*}$ with respect to 
		the modified arc length 
		$\ell_{p}$.
		Let
		$P$ be an $s$-$t$ shortest path having the smallest number of arcs in $G_{\hat{x}^*, \hat{y}^*}$ with respect to the modified arc length 
		$\ell_{p}$.
		\item[Step 3-3:] 
		For $(i,a) \in U$,
		\begin{align*}
		&\hat{x}^*_{(i,a)} \leftarrow
		\begin{cases}
		\hat{x}^*_{(i,a)} - 1 & \text{if $((i,a),(j,b)) \in P \cap E_{\hat{x}^*}$},\\
		\hat{x}^*_{(i,a)} + 1 & \text{if $((j,b),(i,a)) \in P \cap E_{\hat{x}^*}$},\\
		\hat{x}^*_{(i,a)} & \text{otherwise},
		\end{cases}\\
		&\hat{y}^*_{(i,a)} \leftarrow
		\begin{cases}
		\hat{y}^*_{(i,a)} + 1 & \text{if $((i,a),(j,b)) \in P \cap E_{\hat{y}^*}$},\\
		\hat{y}^*_{(i,a)} - 1 & \text{if $((j,b),(i,a)) \in P \cap E_{\hat{y}^*}$},\\
		\hat{y}^*_{(i,a)} & \text{otherwise}.
		\end{cases}
		\end{align*}
		For $v \in V$, $p(v) \leftarrow p(v) + \Delta p(v)$.
		\qed
	\end{description}
\end{description}

At the end of Step~2,
we obtain a minimizer of $\overline{f}|_r$.
The validity of Step~2 is given
in~\cite[Theorem~3.2]{MOR/MS99}.
The time complexity of this algorithm is as follows,
where $n := |U| = \sum_{i \in [r]} d_i$ (the proof is in Section~\ref{sec:proofs}).
\begin{thm}\label{thm:algo}
	The proposed algorithm runs
	in $O(nr^3 + nr \log n + n^2)$ time.
\end{thm}

By improving the algorithm of running time $O(n^3)$ given in~\cite{AI/CZ11},
Cooper--\v{Z}ivn\'{y}~\cite{JAIR/CZ12} gave an $O(n^2 \log n \log r)$-time algorithm for minimizing Z-free functions of the form~(\ref{eq:binary VCSP}).
Our proposed algorithm is faster than Cooper--\v{Z}ivn\'y's for some $r$ (e.g., $r = O(n^{1/3})$).

\begin{remark}\label{rem:oracle vs VCSP}
	\upshape
	In the VCSP framework,
	we assume that the function $f$ of the form~(\ref{eq:binary VCSP}) is explicitly given.
	This means that the input size is proportional to
	\begin{align*}
		\sum_{i \in [r]} d_i + \sum_{i \in [r]} \sum_{d \in D_i} \log c_i(d) + \sum_{1 \leq i < j \leq r} \sum_{d \in D_i} \sum_{e \in D_j} \log c_{ij}(d, e),
	\end{align*}
	and then the running time in Theorem~\ref{thm:algo} is strongly polynomial in the input size.
	On the other hand,
	if we assume that $f$ is given by the value oracles for the functions $c_i$ and $c_{ij}$,
	the input size of $f$ is proportional to
	\begin{align*}
		r + \sum_{i \in [r]} \log d_i + \sum_{i \in [r]} \sum_{d \in D_i} \log c_i(d) + \sum_{1 \leq i < j \leq r} \sum_{d \in D_i} \sum_{e \in D_j} \log c_{ij}(d, e).
	\end{align*}
	In this case,
	the running time in Theorem~\ref{thm:algo} is pseudo-polynomial in the input size.
\end{remark}

\section{Proofs}\label{sec:proofs}
In this section, we give the proofs of Lemma~\ref{lem:Mnatural {0,1}},
Theorem~\ref{thm:Mnatural conv completable}, Theorem~\ref{thm:main}, and Theorem~\ref{thm:algo}.

\paragraph{Proof of Lemma~\ref{lem:Mnatural {0,1}}.}
(only-if part).
Suppose that there exist distinct $i,j,k \in [n]$ such that $h_{ij} < \min \{h_{jk}, h_{ik}\}$.
Note that $h_{ij} < +\infty$ holds.
Then
\begin{align*}
f(\chi_i + \chi_j) + f(\chi_k) = h_i+h_j+h_k + h_{ij} &< h_i+ h_j + h_k + h_{jk} = f(\chi_j + \chi_k) + f(\chi_i),\\
f(\chi_i + \chi_j) + f(\chi_k) = h_i+h_j+h_k + h_{ij} &< h_i+ h_j + h_k + h_{ik} = f(\chi_i + \chi_k) + f(\chi_j)
\end{align*}
hold since $h_i, h_j,h_k,h_{ij} < +\infty$.
By Condition~1 of Theorem~\ref{thm:{0,1}M}, $f$ is not M${}^\natural$-convex.

Suppose that there exist distinct $i,j \in [n]$ such that $h_{ij} < 0 (<+\infty)$.
Then
\begin{align*}
f(\chi_i + \chi_j) + f(\chi_0) = h_i+h_j + h_{ij} < h_i+ h_j = f(\chi_i) + f(\chi_j)
\end{align*}
holds since $h_i, h_j < +\infty$.
By Condition~2 of Theorem~\ref{thm:{0,1}M}, $f$ is not M${}^\natural$-convex.

(if part).
Take arbitrary distinct $i,j,k \in [n]$ and $z \in \{0,1\}^n$ with $\textrm{supp}^+(z) \subseteq [n] \setminus \{i,j,k\}$.
If $f(z+\chi_i+\chi_j) = +\infty$ or $f(z + \chi_k) = +\infty$ holds,
then Condition~1 of Theorem~\ref{thm:{0,1}M} obviously holds.
We assume $f(z+\chi_i+\chi_j) < +\infty$ and $f(z + \chi_k) < +\infty$.

It holds that
\begin{align}
f(z + \chi_i + \chi_j) &= f(z) + h_i + h_j + \sum_{p \in \textrm{supp}^+(z)} h_{ip} + \sum_{p \in \textrm{supp}^+(z)} h_{jp} + h_{ij},\label{eq:f(z+i+j)}\\
f(z + \chi_k) &= f(z) + h_k + \sum_{p \in \textrm{supp}^+(z)} h_{kp}.\label{eq:f(z+k)}
\end{align}
Note that all terms appearing in (\ref{eq:f(z+i+j)}) and (\ref{eq:f(z+k)}) have finite values
since $f(z+\chi_i+\chi_j) < +\infty$ and $f(z + \chi_k) < +\infty$ hold.
Then we have
\begin{align*}
&f(z + \chi_i + \chi_j) + f(z + \chi_k) \geq f(z + \chi_j + \chi_k) + f(z + \chi_i)\\
\Leftrightarrow\ &2f(z) + h_i + h_j + h_k + \sum_{p \in \textrm{supp}^+(z)} h_{ip} + \sum_{p \in \textrm{supp}^+(z)} h_{jp} + \sum_{p \in \textrm{supp}^+(z)} h_{kp} + h_{ij}\\
&\geq 2f(z) + h_j + h_k + h_i + \sum_{p \in \textrm{supp}^+(z)} h_{jp} + \sum_{p \in \textrm{supp}^+(z)} h_{kp} + \sum_{p \in \textrm{supp}^+(z)} h_{ip} + h_{jk}\\
\Leftrightarrow\ &h_{ij} \geq h_{jk}.
\end{align*}
Also we have
\begin{align*}
&f(z + \chi_i + \chi_j) + f(z + \chi_k) \geq f(z + \chi_i + \chi_k) + f(z + \chi_j)\\
\Leftrightarrow\ &h_{ij} \geq h_{ik}.
\end{align*}
By the assumption, it holds that $h_{ij} \geq \min\{h_{jk}, h_{ik}\}$.
Hence we obtain
\begin{align*}
f(z + \chi_i + \chi_j) + f(z + \chi_k) \geq \min\{f(z + \chi_j + \chi_k) + f(z + \chi_i), f(z + \chi_i + \chi_k) + f(z + \chi_j)\}.
\end{align*}

By the assumption of $h_{ij} \geq 0$,
we also obtain
\begin{align*}
f(z + \chi_i + \chi_j) + f(z) \geq f(z + \chi_i) + f(z + \chi_j)
\end{align*}
for all distinct $i, j \in [n]$.

\paragraph{Proof of Theorem~\ref{thm:Mnatural conv completable}.}
First we give a graphical interpretation for the anti-ultrametric property.
For $H := (h_{ij})_{i,j \in [n]}$ and $\alpha \in \overline{\mathbf{R}}$, let us define $E_H^{\alpha}$ and $V_H^{\alpha}$ by
\begin{align}
	E_H^\alpha &:= \{ \{i,j\} \in E_H \mid h_{ij} \geq \alpha \},\label{eq:E_H^alpha}\\
	V_H^\alpha &:= \{ i \mid \exists e \in E_H^\alpha \text{ such that } i \in e \}\label{eq:V_H^alpha}.
\end{align}
Let $G_H^{\alpha} := (V_H^{\alpha},E_H^{\alpha})$.
Then the following lemma holds:
\begin{lem}\label{lem:anti-ultrametric}
	$H := (h_{ij})_{i,j \in [n]}$ satisfies the anti-ultrametric property if and only if
	each connected component of $G_H^{\alpha}$ is a complete graph for every $\alpha \in \overline{\mathbf{R}}$.
\end{lem}
\begin{proof}
	(only-if part).
	We show the contraposition.
	Suppose that for some $\alpha \in \overline{\mathbf{R}}$ there exists a non-complete graph among the connected components of $G_H^{\alpha}$.
	Then there exist distinct $i,j,k \in [n]$ with $\{i,j\}, \{j,k\} \in E_H^{\alpha} \not\ni \{i,k\}$.
	By the definition of $E_H^{\alpha}$,
	it holds that $\min\{h_{ij}, h_{jk}\} \geq \alpha > h_{ik}$.
	This means that $\{h_{ij}, h_{jk}, h_{ik}\}$ does not satisfy the anti-ultrametric property.
	
	(if part).
	Suppose that each connected component of $G_H^{\alpha}$ is a complete graph for all $\alpha \in \overline{\mathbf{R}}$.
	To show the anti-ultrametric property of $(h_{ij})_{i,j \in [n]}$,
	it suffices to prove $h_{jk} = h_{ik}$ for all distinct $i,j,k$ satisfying $h_{ij} > h_{jk}$.
	If $h_{ik} \geq h_{ij}$, then there exists a non-complete graph among the connected components of $G_H^\alpha$ for $\alpha = h_{ij}$, which is a contradiction.
	If $h_{ij} > h_{ik} > h_{jk}$, then there exists a non-complete graph among the connected components of $G_H^\alpha$ for $\alpha = h_{ik}$, which is a contradiction.
	If $h_{jk} > h_{ik}$, then there exists a non-complete graph among the connected components of $G_H^\alpha$ for $\alpha = h_{jk}$, which is a contradiction.
	Therefore we must have $h_{jk} = h_{ik}$.
\end{proof}

We are now ready to prove Theorem~\ref{thm:Mnatural conv completable}.
\begin{proof}[Proof of Theorem~\ref{thm:Mnatural conv completable}]
	(only-if part).
	Suppose to the contrary that $H := (h_{ij})_{i,j \in [n]}$ is M${}^\natural$-convex completable
	and that there exists a chordless cycle $C$ of $G_H$ with $|\argmin_{e \in C} w(e)| = 1$.
	Let
	$C=  \{ \{i_1,i_2\}, \{i_2,i_3\}, \dots, \{i_m,i_1\} \}$,
	and consider the corresponding  entries
	$\{h_{i_1i_2}, h_{i_2i_3}, \dots, h_{i_mi_1}\}$  of $H$.
	Note that $h_{i_pi_q}$ is undefined for $p,q \in [m]$ with 
	$|p-q| \not= 1 \mod m$.
	We may assume $\alpha := h_{i_1i_2} = \min\{h_{i_1i_2}, h_{i_2i_3}, \dots, h_{i_mi_1}\}$.
	By the assumption of $|\argmin_{e \in C} w(e)| = 1$, we have
	$\min\{h_{i_2i_3}, \dots, h_{i_mi_1}\} > \alpha$.
	Since
	$\{h_{i_1i_2}, h_{i_2i_3}, h_{i_1i_3}\}$ 
	should satisfy
	the anti-ultrametric property, we have to assign $\alpha$ to $h_{i_1i_3}$ to obtain an M${}^\natural$-convex completion.
	Since $\{h_{i_1i_3}, h_{i_3i_4}, h_{i_1i_4}\}$ should satisfy the anti-ultrametric property,
	we have to assign $\alpha$ to $h_{i_1i_3}$ to obtain an M${}^\natural$-convex completion.
	By repeating this procedure, we arrive at $h_{i_1i_{m-1}} = \alpha$.
	This is a contradiction, since
	$h_{i_1i_{m-1}} < \min\{h_{i_{m-1}i_m}, h_{i_1i_m}\}$
	and hence the anti-ultrametric property fails
	for  $\{h_{i_1i_{m-1}}, h_{i_{m-1}i_m}, h_{i_mi_1}\}$.
	
	(if part).
	For $\alpha \in \overline{\mathbf{R}}_+$,
	define $S_H^{\alpha}$ by
	\begin{align*}
	S_H^\alpha &:= \{ \{ i,j \} \not\in E_H \mid i,j \in V_H^\alpha,\ \text{$i$ and $j$ are connected in $G_H^\alpha$} \}.
	\end{align*}
	Let $\tilde{G}_H^{\alpha} := (V_H^{\alpha}, E_H^{\alpha} \cup S_H^{\alpha})$.
	Recall that $E_H^{\alpha}$ and $V_H^{\alpha}$ are defined in (\ref{eq:E_H^alpha}) and (\ref{eq:V_H^alpha}).
	
	First we show that if each connected component of $\tilde{G}_H^\alpha$ is a complete graph for every $\alpha \in \overline{\mathbf{R}}_+$,
	then $H$ is M${}^\natural$-convex completable.
	Let $\alpha_1 > \alpha_2 > \cdots > \alpha_p$ be the distinct values of defined elements of $(h_{ij})_{i,j \in [n]}$ ($\alpha_1$ can be the infinite value).
	We assign $\alpha_1$ to each undefined element $h_{ij}$ such that $\{i,j\} \in S_H^{\alpha_1}$,
	$\alpha_k$ to each $h_{ij}$ such that $\{i,j\} \in S_H^{\alpha_k} \setminus S_H^{\alpha_{k-1}}$ for $k = 2, \dots, p-1$,
	and $\alpha_p$ to each $h_{ij}$ such that $\{i,j\} \not\in S_H^{\alpha_{p-1}}$.
	Then we obtain a certain completion $\tilde{H} := (\tilde{h}_{ij})_{i,j \in [n]}$ of $(h_{ij})_{i,j \in [n]}$.
	It is clear that each connected component of $G_{\tilde{H}}^\alpha$ is a complete graph for every $\alpha \in \overline{\mathbf{R}}$.
	By Lemma~\ref{lem:anti-ultrametric}, $\tilde{H}$ satisfies the anti-ultrametric property.
	This means that $H$ is M${}^\natural$-convex completable.
	
	Next we show that if $|\argmin_{e \in C} w(e) | \geq 2$ holds for every chordless cycle $C$ of $G_H$,
	then each connected component of $\tilde{G}_H^\alpha$ is a complete graph for every $\alpha \in \overline{\mathbf{R}}_+$.
	Take arbitrary $\alpha \in \overline{\mathbf{R}}_+$ and $i$ and $j$ which are connected in $\tilde{G}_H^\alpha$
	(Note that vertex sets of connected components of $\tilde{G}_H^\alpha$ are the same as those of $G_H^\alpha$).
	It suffices to prove that $\{i,j\} \in E_H^{\alpha}$ or $\{i,j\} \in S_H^{\alpha}$ holds.
	Suppose to the contrary that there exist $i$ and $j$ such that $\{i,j\} \not\in E_H^{\alpha}$ and $\{i,j\} \not\in S_H^{\alpha}$ hold.
	Let $I$ be the set of such $\{i,j\}$.
	Let $\{i_0, j_0\} \in I$ be a pair of vertices such that the number of edges of a shortest $i_0$-$j_0$ path on $G_H^\alpha$ is minimum in $I$.
	Since $i_0$ and $j_0$ are connected in $G_H^\alpha$ and $\{i_0, j_0\} \not\in S_H^\alpha$,
	we have $\{i_0, j_0\} \in E_H$.
	Moreover since $\{i_0, j_0\} \not\in E_H^\alpha$,
	$h_{i_0j_0} < \alpha$ holds.
	Take a $i_0$-$j_0$ shortest path $P_0$.
	Then $P_0 \cup \{i_0, j_0\}$ is a chordless cycle of $G_H$.
	Indeed, if $P_0 \cup \{i_0, j_0\}$ has a chord in $G_H$,
	there exist $i'$ and $j'$ satisfying $\{i',j'\} \neq \{i_0,j_0\}$ in $P_0 \cup \{i_0, j_0\}$ such that $E_H^\alpha \not\ni \{i', j'\} \in E_H$ by the minimality of $|I|$.
	Then $\{i', j'\} \in I$ and the number of edges of a shortest $i'$-$j'$ path is smaller that those of $P_0$.
	However this is a contradiction to the minimality of $\{i_0, j_0\}$.
	Hence $P_0 \cup \{i_0, j_0\}$ is a chordless cycle of $G_H$.
	It holds that $h_{ij} \geq \alpha$ for $\{i,j\} \in P_0$ and $h_{i_0j_0} < \alpha$.
	Therefore we obtain $|\argmin_{e \in P_0 \cup \{i_0,j_0\}} w(e)| = 1$.
	This contracts the assumption of $|\argmin_{e \in C} w(e) | \geq 2$.
	Hence we have $\{i,j\} \in E_H^{\alpha}$ or $\{i,j\} \in S_H^{\alpha}$.
	\end{proof}
	
	\paragraph{Proof of Theorem~\ref{thm:main}.}
	Let $H := (h_{(i,a), (j,b)})_{(i,a), (j,b) \in U}$,
	where the entries $h_{(i,a), (j,b)}$ with $i=j$ are undefined.
	Recall that $G_H = (U, E_H; w)$ is the assignment graph of $H$.
	By the definition of $E_H$ and $h_{(i,a), (j,b)}$ in (\ref{eq:h_(i,a)(j,b)}),
	we have $E_H = \{ \{(i,a), (j,b)\} \mid i \neq j,\ a \in D_i,\ b \in D_j \}$.
	By Theorem~\ref{thm:Mnatural conv completable}, $H$ is
	M${}^\natural$-convex completable if and only if 
	every chordless cycle $C$ satisfies the condition  $|\argmin_{e \in C} w(e)| \geq 2$. 
	
	First we show that  chordless cycles in $G_H $ have length 3 or 4.
	Take any chordless cycle $C = \{ \{(i_1, a_1), (i_2,a_2)\}, \{(i_2, a_2), (i_3,a_3)\}, \dots, \{(i_k, a_k), (i_1,a_1)\} \}$ of $G_H$.
	Since $C$ is chordless,
	we have $i_1 = i_p$ for $3 \leq p \leq k-1$ and $i_2 = i_q$ for $4 \leq q \leq k$.
	This implies $k \leq 4$, since otherwise we obtain $i_1 = i_4 = i_2$,
	contradicting the existence of  an edge between $(i_1,a_1)$ and $(i_2,a_2)$.
	
	For a (chordless) cycle of length 3, say,
	$C = \{ \{(i_1,a_1), (i_2,a_2)\}, \{(i_2,a_2), (i_3,a_3)\}, \{(i_3,a_3), (i_1,a_1)\} \}$
	with $i_1 \neq i_2 \neq i_3 \neq i_1$, 
	the condition $|\argmin_{e \in C} w(e)| \geq 2$ 
	is equivalent to  (\ref{eq:JWP}) for JWP.
	For a chordless cycle of length 4, say,
	$C = \{ \{(i_1,a_1), (i_2,a_2)\}, \{(i_2,a_2), (i_3,a_3)\}, \{(i_3,a_3), (i_4,a_4)\}, \{(i_4,a_4), (i_1,a_1)\} \}$
	we have $i_1 \neq i_2$, $i_3 \neq i_4$, $i_1 = i_3$, $i_2 = i_4$, $a_1 \neq a_3$, $a_2 \neq a_4$, and then 
	the condition $|\argmin_{e \in C} w(e)| \geq 2$ 
	is equivalent to (\ref{eq:Z-free}) for Z-freeness.

\paragraph{Proof of Theorem~\ref{thm:algo}.}
We investigate each step 
in turn.

(Step 1).
Since the number of defined elements of $(h_{(i,a),(j,b)})_{(i,a),(j,b) \in U}$ is $O(n^2 - \sum_{i=1}^r d_i^2) = O(n^2)$, we can find an M${}^\natural$-convex completion in $O(n^2 + n \log n)$ time (recall Remark~\ref{rem:completion}).

(Step 2).
If we have the value of $\overline{f}(\hat{x}^*)$,
we can compute the value of $\overline{f}(\hat{x}^* + \chi_{(i,a)})$ in $O(r)$ time
since $\overline{f}(\hat{x}^* + \chi_{(i,a)}) = \overline{f}(\hat{x}^*) + c_i(a) + \sum_{(j,b) \in \textrm{supp}^+(\hat{x}^*)} \tilde{h}_{(i,a), (j,b)}$.
Hence the time complexity of Step~3 is $O(nr^2)$ time.

(Step 3).
Recall the definition of $G_{\hat{x}^*, \hat{y}^*}$ in (\ref{eq:V})--(\ref{eq:w}).
We have $|E_{\hat{x}^*}| = O(r(n-r)) = O(nr)$,
$|E_{\hat{y}^*}| = O(n)$,
$|E^+| = O(r)$,
and $|E^-| = O(r)$.
Hence $|E_{\hat{x}^*, \hat{y}^*}| = O(nr)$.
Furthermore we 
need to compute $\ell$
only on $E_{\hat{x}^*}$,
since $\ell$ is equal to zero on other arcs.
If we have the value of $\overline{f}(\hat{x}^*)$
at hand,
we can compute the value of $\overline{f}(\hat{x}^* + \chi_{(j,b)} - \chi_{(i,a)})$ in $O(r)$ time
since
\begin{align*}
&\overline{f}(\hat{x}^* + \chi_{(j,b)} - \chi_{(i,a)})\\
&= \overline{f}(\hat{x}^*) - \left(c_i(a) + \sum_{(k,c) \in \textrm{supp}^+(\hat{x}^*)} \tilde{h}_{(i,a), (k,c)}\right) + \left(c_j(b) + \sum_{(k,c) \in \textrm{supp}^+(\hat{x}^* - \chi_{(i,a)})} \tilde{h}_{(j,b), (k,c)}\right).
\end{align*}
Therefore we can
construct
the auxiliary graph $G_{\hat{x}^*, \hat{y}^*}$ 
in $O(nr^2)$ time.

The modified arc length 
$\ell_{p}$
is nonnegative~\cite[Section~5.2]{book/Murota00}.
Hence we can compute $\Delta p(v)$ for $v \in V$ 
and a shortest path $P$ in 
Step~3-2 
in $O(nr + n \log n)$ time by using Dijkstra's algorithm with Fibonacci heaps~\cite{JACM/FT87} (see also~\cite[Section~7.4]{book/Schrijver03}).
We can update $\hat{x}^*$, $\hat{y}^*$, and $p$ in 
Step~3-3 
in $O(nr)$ time.
By one iteration of 
	Step~3,
	the value of $\| \hat{x}^* - \hat{y}^* \|_1$ 
	is decreased by two.
Hence the number of iterations of 
Step~3
is bounded by $O(r)$.
Therefore the time complexity of 
Step~3 
is $O(nr^3 + nr \log n)$.

By the above argument,
we see
that 
the proposed algorithm 
runs
in $O(nr^3 + nr \log n + n^2)$ time.

\section*{Acknowledgments}
We thank Kazutoshi Ando and Takanori Maehara for information on the paper~\cite{Algo/FKW95} in Remark~\ref{rem:completion}.
We also thank the referees for helpful comments.
This research was initiated at the Trimester Program ``Combinatorial Optimization''
at Hausdorff Institute of Mathematics, 2015.
The first author's research was supported by JSPS Research Fellowship for Young Scientists.
The second author's research was supported by The Mitsubishi Foundation, CREST, JST,
and JSPS KAKENHI Grant Number 26280004.
The last author's research was supported by a Royal Society University 
Research Fellowship. This project has received funding from the European 
Research Council (ERC) under the European Union's Horizon 2020 research 
and innovation programme (grant agreement No 714532). The paper reflects 
only the authors' views and not the views of the ERC or the European 
Commission. The European Union is not liable for any use that may be 
made of the information contained therein.


\end{document}